\documentclass[pra,aps,nopacs,onecolumn,twoside,superscriptaddress]{revtex4}
\usepackage{amsmath,amsfonts,amssymb,caption,color,epsfig,graphics,graphicx,hyperref,latexsym,mathrsfs,revsymb,theorem,url,verbatim,epstopdf,mathtools,enumerate,stmaryrd,enumitem,appendix,multirow,parskip}

\hypersetup{colorlinks,linkcolor={blue},citecolor={blue},urlcolor={red}}

\newtheorem{definition}{Definition}
\newtheorem{proposition}[definition]{Proposition}
\newtheorem{lemma}[definition]{Lemma}

\newtheorem{theorem}[definition]{Theorem}
\newtheorem{corollary}[definition]{Corollary}
\newtheorem{conjecture}[definition]{Conjecture}

\newtheorem{remark}[definition]{Remark}
\newtheorem{example}[definition]{Example}
\newtheorem{question}[definition]{Question}
\newtheorem{memo}[definition]{Memo}

\def\squareforqed{\hbox{\rlap{$\sqcap$}$\sqcup$}}
\def\qed{\ifmmode\squareforqed\else{\unskip\nobreak\hfil
\penalty50\hskip1em\null\nobreak\hfil\squareforqed
\parfillskip=0pt\finalhyphendemerits=0\endgraf}\fi}
\def\endenv{\ifmmode\;\else{\unskip\nobreak\hfil
\penalty50\hskip1em\null\nobreak\hfil\;
\parfillskip=0pt\finalhyphendemerits=0\endgraf}\fi}
\newenvironment{proof}{\noindent \textbf{{Proof.~} }}{\qed}
\def\Dbar{\leavevmode\lower.6ex\hbox to 0pt
{\hskip-.23ex\accent"16\hss}D}
\makeatletter
\def\url@leostyle{%
  \@ifundefined{selectfont}{\def\UrlFont{\sf}}{\def\UrlFont{\small\ttfamily}}}
\makeatother
\def\bcj{\begin{conjecture}}
\def\ecj{\end{conjecture}}
\def\bcr{\begin{corollary}}
\def\ecr{\end{corollary}}
\def\bd{\begin{definition}}
\def\ed{\end{definition}}
\def\bea{\begin{eqnarray}}
\def\eea{\end{eqnarray}}
\def\beq{\begin{equation}}
\def\eeq{\end{equation}}
\def\bal{\begin{aligned}}
\def\eal{\end{aligned}}
\def\bem{\begin{enumerate}}
\def\eem{\end{enumerate}}
\def\bex{\begin{example}}
\def\eex{\end{example}}
\def\bim{\begin{itemize}}
\def\eim{\end{itemize}}
\def\bl{\begin{lemma}}
\def\el{\end{lemma}}
\def\bma{\begin{bmatrix}}
\def\ema{\end{bmatrix}}
\def\bpf{\begin{proof}}
\def\epf{\end{proof}}
\def\bpp{\begin{proposition}}
\def\epp{\end{proposition}}
\def\bqu{\begin{question}}
\def\equ{\end{question}}
\def\br{\begin{remark}}
\def\er{\end{remark}}
\def\bt{\begin{theorem}}
\def\et{\end{theorem}}
\def\bmm{\begin{memo}}
\def\emm{\end{memo}}

\def\btb{\begin{tabular}}
\def\etb{\end{tabular}}

\newcommand{\nc}{\newcommand}

\def\l{\lambda}

\def\s{\sigma}

\nc{\bbA}{\mathbb{A}} \nc{\bbB}{\mathbb{B}} \nc{\bbC}{\mathbb{C}}
 \nc{\bbD}{\mathbb{D}} \nc{\bbE}{\mathbb{E}} \nc{\bbF}{\mathbb{F}}
 \nc{\bbG}{\mathbb{G}} \nc{\bbH}{\mathbb{H}} \nc{\bbI}{\mathbb{I}}
 \nc{\bbJ}{\mathbb{J}} \nc{\bbK}{\mathbb{K}} \nc{\bbL}{\mathbb{L}}
 \nc{\bbM}{\mathbb{M}} \nc{\bbN}{\mathbb{N}} \nc{\bbO}{\mathbb{O}}
 \nc{\bbP}{\mathbb{P}} \nc{\bbQ}{\mathbb{Q}} \nc{\bbR}{\mathbb{R}}
 \nc{\bbS}{\mathbb{S}} \nc{\bbT}{\mathbb{T}} \nc{\bbU}{\mathbb{U}}
 \nc{\bbV}{\mathbb{V}} \nc{\bbW}{\mathbb{W}} \nc{\bbX}{\mathbb{X}}
 \nc{\bbZ}{\mathbb{Z}}

 \nc{\bA}{{\bf A}} \nc{\bB}{{\bf B}} \nc{\bC}{{\bf C}}
 \nc{\bD}{{\bf D}} \nc{\bE}{{\bf E}} \nc{\bF}{{\bf F}}
 \nc{\bG}{{\bf G}} \nc{\bH}{{\bf H}} \nc{\bI}{{\bf I}}
 \nc{\bJ}{{\bf J}} \nc{\bK}{{\bf K}} \nc{\bL}{{\bf L}}
 \nc{\bM}{{\bf M}} \nc{\bN}{{\bf N}} \nc{\bO}{{\bf O}}
 \nc{\bP}{{\bf P}} \nc{\bQ}{{\bf Q}} \nc{\bR}{{\bf R}}
 \nc{\bS}{{\bf S}} \nc{\bT}{{\bf T}} \nc{\bU}{{\bf U}}
 \nc{\bV}{{\bf V}} \nc{\bW}{{\bf W}} \nc{\bX}{{\bf X}}
 \nc{\bZ}{{\bf Z}}

\nc{\cA}{{\cal A}} \nc{\cB}{{\cal B}} \nc{\cC}{{\cal C}}
\nc{\cD}{{\cal D}} \nc{\cE}{{\cal E}} \nc{\cF}{{\cal F}}
\nc{\cG}{{\cal G}} \nc{\cH}{{\cal H}} \nc{\cI}{{\cal I}}
\nc{\cJ}{{\cal J}} \nc{\cK}{{\cal K}} \nc{\cL}{{\cal L}}
\nc{\cM}{{\cal M}} \nc{\cN}{{\cal N}} \nc{\cO}{{\cal O}}
\nc{\cP}{{\cal P}} \nc{\cQ}{{\cal Q}} \nc{\cR}{{\cal R}}
\nc{\cS}{{\cal S}} \nc{\cT}{{\cal T}} \nc{\cU}{{\cal U}}
\nc{\cV}{{\cal V}} \nc{\cW}{{\cal W}} \nc{\cX}{{\cal X}}
\nc{\cY}{{\cal Y}} \nc{\cZ}{{\cal Z}}

\nc{\hA}{{\hat{A}}} \nc{\hB}{{\hat{B}}} \nc{\hC}{{\hat{C}}}
\nc{\hD}{{\hat{D}}} \nc{\hE}{{\hat{E}}} \nc{\hF}{{\hat{F}}}
\nc{\hG}{{\hat{G}}} \nc{\hH}{{\hat{H}}} \nc{\hI}{{\hat{I}}}
\nc{\hJ}{{\hat{J}}} \nc{\hK}{{\hat{K}}} \nc{\hL}{{\hat{L}}}
\nc{\hM}{{\hat{M}}} \nc{\hN}{{\hat{N}}} \nc{\hO}{{\hat{O}}}
\nc{\hP}{{\hat{P}}} \nc{\hR}{{\hat{R}}} \nc{\hS}{{\hat{S}}}
\nc{\hT}{{\hat{T}}} \nc{\hU}{{\hat{U}}} \nc{\hV}{{\hat{V}}}
\nc{\hW}{{\hat{W}}} \nc{\hX}{{\hat{X}}} \nc{\hZ}{{\hat{Z}}}

\nc{\hn}{{\hat{n}}}

\def\max{\mathop{\rm max}}

\def\tr{\mathop{\rm Tr}}

\def\ox{\otimes}

\newcommand{\braket}[2]{\langle#1|#2\rangle}

\newcommand{\norm}[1]{\lVert#1\rVert}

\begin{document}

\title{On the distillablity conjecture in matrix theory}

\author{Saiqi Liu}\email[]{liu\_saiqi\_sy2324111@buaa.edu.cn}
\affiliation{LMIB(Beihang University), Ministry of education, and School of Mathematical Sciences, Beihang University, Beijing 100191, China}

\author{Lin Chen}\email[]{linchen@buaa.edu.cn (corresponding author)}
\affiliation{LMIB(Beihang University), Ministry of education, and School of Mathematical Sciences, Beihang University, Beijing 100191, China}

\date{\today}



\begin{abstract}
The distillability conjecture of two-copy 4 × 4 Werner states is one of the main open
problems in quantum information. We prove two special cases of the conjecture. The first case occurs when two 4 × 4 matrices $A$, $B$ are both unitarily equivalent to block diagonal matrices with $2$ by $2$ blocks. The second case occurs when $B$ is unitarily equivalent to either $-A$ or $-A^T$. Plus, we propose a simplified version of the distillability conjecture when both $A$ and $B$ are matrices with distinct eigenvalues. 
\end{abstract}

\maketitle


\section{Introduction}

Entanglement distillability conjecture is one of the five key theoretical problems in quantum information \cite{2020Five}. Basically, the conjecture asks whether all non-positive-partial-transpose (NPT) entangled states can be locally converted into pure states asymptotically. This is known to be feasible for any two-qubit entanglement \cite{hhh97}. Many mathematical tools have been developed to study the distillability problem in the past decades, such as the reduction criterion and entanglement witness for activation \cite{hh1999,klc02}. The quantification of the upper bound for distillable entanglement has also been evaluated \cite{rains1999,rains2001}. Next, the partial transpose of bipartite entangled states has been applied to entanglement distillation \cite{dcl00}. The distillability problem can also be considered by using several copies of target states \cite{br03}. On the other hand, distilling entanglement has been applied to the distillation of security keys \cite{dw2005}. The distillability problem has also been tested by semidefinite programming  \cite{vd06}. It was proven that the relation between local ranks and the global rank of a bipartite state is intimately related to the distillability of this state \cite{cd11jpa}. Next, bipartite NPT states of rank at most four have been shown to be distillable \cite{cc08,cd16pra}. The distillability problem is also related to a matrix inequality from zero-entropy dependent on the rank of states \cite{1909.02748}. 

\vspace{7pt}
A special case of distillability conjecture, namely the distillability of two-copy $4\times4$ Werner states has received a lot of attention and progress in recent years
\cite{Shen2018On,QIAN2021139,2021Proving,2023On}. In this paper, we shall further study this two-copy problem by solving some special cases in terms of matrix theory. We establish two specific instances of Conjecture \ref{conjecture: distillability}. The first instance occurs when both matrices \( A \) and \( B \) are unitarily equivalent to block-diagonal matrices composed of \( 2 \times 2 \) blocks. This case is formalized in Theorem \ref{thm: 2 2 bloc-diag} and is demonstrated through Proposition \ref{prop: equiv prob 2 2 bloc-diag} and Lemmas \ref{lemma: 1st case 2 2 bloc-diag} to \ref{lemma: 3rd case 2 2 bloc-diag}. The second instance arises under the conditions where \( B \) is unitarily equivalent to either \( -A \) or \( -A^T \), as presented in Theorem \ref{thm: B = -A^T or -A} and proven through Lemmas \ref{lemma: B = -A^T} and \ref{lemma: B = -A}. These cases provide novel insights into the structural properties of the conjecture and indicate potential approaches for its proof in more general contexts by revealing underlying symmetries and constraints that may extend to broader classes of matrices. In particular, in the second instance, we establish a link between the distillability conjecture and the Frobenius norm of commutators. Furthermore, by exploiting some topological properties, namely density, compactness, and continuity, a simplified version of the distillability conjecture is proposed in Theorem \ref{thm: diag prob set}, which is supported by Lemmas \ref{lemma: phi} and \ref{lemma: dense}. We prove that the distillability conjecture holds if and only if it holds for matrices $A$ and $B$ with distinct eigenvalues.

\vspace{7pt}
The rest of this paper is organized as follows. We review some necessary notions and facts in Sec. \ref{sec:pre}. Then we introduce our results in Sec. \ref{sec:res}, and finally we conclude in Sec. \ref{sec:con}.

\section{Preliminaries}\label{sec:pre}

In this section, we introduce certain notions and facts used in the derivation of this paper. Most of them can be found in \cite{Horn_Johnson_1991}.
\vspace{7pt}
We use $*$ to represent a complex entry of a matrix. For example, $F$ defined as follows is a $2$ by $2$ complex matrix with zero diagonal entries
\begin{align}
    F = \begin{bmatrix}
        0 & * \\
        * & 0
    \end{bmatrix}.
\end{align}

\begin{definition}\label{def: vec}
    The set of all $m$ by $n$ matrices over the complex field $\bbC$ is denoted by $M_{m,n}(\bbC)$, and $M_{n,n}(\bbC)$ is denoted by $M_n(\bbC)$.
\end{definition}

\begin{definition}
    With each matrix $A = [a_{ij}] \in M_{m,n}(\bbC)$, we associate the vector $\operatorname{vec} A$ defined by
    \beq \bal
        \operatorname{vec} A := [a_{11}, \cdots, a_{m1}, a_{12}, \cdots, a_{m2}, \cdots, a_{1n}, \cdots a_{mn}]^T.
    \eal \eeq
\end{definition}

\begin{definition}
    The Kronecker product of $A = [a_{ij}] \in M_{m,n}$ and $B = [b_{ij}] \in M_{p,q}$ is denoted by $A \otimes B$ and is defined to be the block matrix
    \beq \bal
        A \otimes B := \bma
            a_{11} B & \cdots & a_{1n} B \\
            \vdots & \ddots & \vdots \\
            a_{m1} B & \cdots & a_{mn} B
        \ema \in M_{mp, nq}.
    \eal \eeq
\end{definition}

\begin{definition}
    The Frobenius norm of matrix $A$ is denoted by $\| A \|_F$ and is defined by
    \beq \bal
        \| A \|_F := \left(\tr \left(A^* A\right)\right)^{\frac{1}{2}}.
    \eal \eeq
    Plus, the standard inner product of matrices $A$ and $B$ is denoted by $\braket{A}{B}$ and is defined by
    \beq \bal
        \braket{A}{B} := \tr (A^* B).
    \eal \eeq
\end{definition}

\begin{definition}\label{def: 2 norm}
    Let $\norm{\cdot}_{\max}$ be a norm defined on  $M_4(\mathbb{C})\times M_4(\mathbb{C})$ as follows
    \begin{equation}
        \begin{aligned}
            \norm{\cdot}_2: (A, B) \mapsto \norm{(A, B)}_2 = \sqrt{\norm{A}_F^2 + \norm{B}_F^2}.
        \end{aligned}
    \end{equation}
\end{definition}

The following lemma is demonstrated on page 255 of \cite{Horn_Johnson_1991}.

\begin{lemma}\label{lemma: krom sum and commutators}
     For all $A, B, C, D \in M_n(\bbC)$
    \beq \bal
        B D + D A^T = C \iff (A \ox I + I \ox B) \operatorname{vec} D = \operatorname{vec} C.
    \eal \eeq
\end{lemma}

Some results on the upper bound of the Frobenius norm of commutators and its variants are also presented. For more details, please refer to \cite{Vong2008ProofOB} and \cite{FONG20111193}.

\begin{lemma}\label{lemma: frob norm xy-yx}
    For all $Y, Z \in M_n(\bbC)$
    \begin{equation}
        \norm{YZ - ZY}_F \leq \sqrt 2 \norm{Y}_F \norm{Z}_F, \label{xy-yx}
    \end{equation}
    \begin{equation}
        \norm{YZ - ZY^T}_F \leq \sqrt 2 \norm{Y}_F \norm{Z}_F. \label{xy-yxt}
    \end{equation}
\end{lemma}

Now let us recapitulate the distillability conjecture.

\begin{conjecture}\label{conjecture: distillability}
    For matrices $A, B \in M_4(\mathbb{C})$  which satisfy
    \begin{align}
        \tr(A) = \tr(B) = 0 \label{tr=0 cond}, \\
        \|A\|_F^2 + \|B\|_F^2 = \frac{1}{4}, \label{frobenius norm condition}
    \end{align}
    we define matrix $X \in M_{16}(\mathbb{C})$
    \beq \bal
        X = A \otimes I + I \otimes B. 
    \eal \eeq
    Then we have
    \beq \bal
        \sigma_1(X)^2 + \sigma_2(X)^2 \leq \frac{1}{2},
    \eal \eeq
    where $\sigma_1(X)$ and $\sigma_2(X)$ are respectively the largest and second-largest singular values of $X$.
\end{conjecture}

There has already been notable progress on the distillability conjecture. In \cite{Shen2018On} and \cite{2021Proving}, the case of monomial matrices $A$ and $B$ was proven. In \cite{2023On}, the case where $A$ and $B$ are matrices with at most four nonzero entries was proven, and in \cite{QIAN2021139}, the case stated in Lemma \ref{lemma: normal A or B} was proven.
\begin{lemma}\label{lemma: normal A or B}
    Conjecture \ref{conjecture: distillability} holds when one of $A$ and $B$ is normal.
\end{lemma}
The proof of this lemma leverages the property that a normal matrix $A$ or $B$ enables the eigenspace of matrix $X$ to be decomposed into four mutually orthogonal eigenspaces, rendering matrix $X$ unitarily equivalent to a block-diagonal matrix with four $4$ by $4$ blocks on the diagonal. This approach inspired us to investigate the case where both $A$ and $B$ are unitarily equivalent to block-diagonal matrices with two $2$ by $2$ blocks on the diagonal. Under this condition, matrix $X$ is also unitarily equivalent to a block-diagonal matrix with four $4$ by $4$ blocks on the diagonal.

\section{Main results}\label{sec:res}

In this section, we present the main results. In Sec. \ref{sec:UEtoBD2x2}, we show that Conjecture \ref{conjecture: distillability} holds when $A$ and $B$ are unitarily equivalent to block-diagonal matrices, each with two $2$ by $2$ blocks. In Sec. \ref{sec:UE=-ATor-A}, we show that Conjecture \ref{conjecture: distillability} holds when matrix $B$ is unitarily equivalent to either $-A^T$ or $-A$. In Sec. \ref{subsec: diag}, we show that Conjecture \ref{conjecture: distillability} holds if and only if it holds for matrices with distinct eigenvalues. 

\subsection{$A$ and $B$ are unitarily equivalent to block-diagonal matrices with two $2$ by $2$ blocks}\label{sec:UEtoBD2x2}

In this subsection, we shall assume that matrices $A$ and $B$ are block-diagonal matrices with two $2$ by $2$ blocks. The main conclusion of this subsection is Theorem \ref{thm: 2 2 bloc-diag}, which is proven using Propositions \ref{prop: equiv prob 2 2 bloc-diag} and Lemmas \ref{lemma: 1st case 2 2 bloc-diag} to \ref{lemma: 3rd case 2 2 bloc-diag}. We start by simplifying the statement of Theorem \ref{thm: 2 2 bloc-diag} with the help of Proposition \ref{prop: equiv prob 2 2 bloc-diag}, \eqref{tr=0 cond}, and \eqref{frobenius norm condition}, resulting in three distinct cases, \ref{1st case: 2 2 bloc-diag}, \ref{2nd case: 2 2 bloc-diag}, and \ref{3rd case: 2 2 bloc-diag}. We then perform a case-by-case analysis of these three cases, simplifying each one through analysis of entries of  $A$ and $B$. In Lemmas \ref{lemma: 1st case 2 2 bloc-diag}, \ref{lemma: 2nd case 2 2 bloc-diag}, and \ref{lemma: 3rd case 2 2 bloc-diag}, we prove the validity of Conjecture \ref{conjecture: distillability} corresponding to each case. Consequently, Theorem \ref{thm: 2 2 bloc-diag} is established.

\begin{theorem}\label{thm: 2 2 bloc-diag}
    Conjecture \ref{conjecture: distillability} holds for matrices $A$ and $B$ that are unitarily equivalent to block-diagonal matrices with two $2$ by $2$ blocks.
\end{theorem}

The proof of Theorem \ref{thm: 2 2 bloc-diag} starts from here, the following proposition is obvious from the fact that singular values are unitarily invariant.

\begin{proposition}\label{prop: equiv prob 2 2 bloc-diag}
    If Conjecture \ref{conjecture: distillability} holds for matrices $A$ and $B$, then Conjecture \ref{conjecture: distillability} holds for matrices $\tilde A$ and $\tilde B$ that are unitarily equivalent to $A$ and $B$.
\end{proposition}

Without loss of generality, we can assume that $A$ and $B$ are block-diagonal matrices with two $2$ by $2$ blocks given by
\beq \bal
    A = 
    \begin{bmatrix}
        A_{11} & 0 \\
        0 & A_{22}
    \end{bmatrix}
    =
    \begin{bmatrix}
        a_{11} & a_{12} & 0 & 0 \\
        a_{21} & a_{22} & 0 & 0 \\
        0 & 0 & a_{33} & a_{34} \\
        0 & 0 & a_{43} & a_{44}
    \end{bmatrix},
    \quad
    B = 
    \begin{bmatrix}
        B_{11} & 0 \\
        0 & B_{22}
    \end{bmatrix}
    =
    \begin{bmatrix}
        b_{11} & b_{12} & 0 & 0 \\
        b_{21} & b_{22} & 0 & 0 \\
        0 & 0 & b_{33} & b_{34} \\
        0 & 0 & b_{43} & b_{44}
    \end{bmatrix}.
\eal \eeq
In order to take advantage of the traceless condition \eqref{tr=0 cond}, we can always obtain $A_{11}, A_{22}, B_{11}$, and $B_{22}$ with equal diagonal entries by unitary equivalence (for more details, please refer to example 2.2.3 in \cite{Horn_Johnson_2012}), i.e. $a = a_{11} = a_{22} = -a_{33} = -a_{44}$ and $b = b_{11} = b_{22} = -b_{33} = -b_{44}$. By Proposition \ref{prop: equiv prob 2 2 bloc-diag}, we know that unitary equivalence does not influence the validity of Theorem \ref{thm: 2 2 bloc-diag}. Now we have
\beq \bal
    A = 
    \begin{bmatrix}
        A_{11} & 0 \\
        0 & A_{22}
    \end{bmatrix}
    =
    \begin{bmatrix}
        a & a_{12} & 0 & 0 \\
        a_{21} & a & 0 & 0 \\
        0 & 0 & -a & a_{34} \\
        0 & 0 & a_{43} & -a
    \end{bmatrix},
    \quad
    B = 
    \begin{bmatrix}
        B_{11} & 0 \\
        0 & B_{22}
    \end{bmatrix}
    =
    \begin{bmatrix}
        b & b_{12} & 0 & 0 \\
        b_{21} & b & 0 & 0 \\
        0 & 0 & -b & b_{34} \\
        0 & 0 & b_{43} & -b
    \end{bmatrix}.
\eal \eeq
By calculations, we have
\beq \bal
    X &= A \otimes I_4 + I_4 \otimes B = 
    \bma
        X_{11} & 0 \\
        0 & X_{22}
    \ema,
\eal \eeq
where
\beq \bal
    X_{11} = 
    \begin{bmatrix}
        b + a & b_{12} & 0 & 0 & a_{12} & 0 & 0 & 0 \\
        b_{21} & b + a & 0 & 0 & 0 & a_{12} & 0 & 0 \\
        0 & 0 & - b + a & b_{34} & 0 & 0 & a_{12} & 0 \\
        0 & 0 & b_{43} & - b + a & 0 & 0 & 0 & a_{12} \\
        a_{21} & 0 & 0 & 0 & b + a & b_{12} & 0 & 0 \\
        0 & a_{21} & 0 & 0 & b_{21} & b + a & 0 & 0 \\
        0 & 0 & a_{21} & 0 & 0 & 0 & - b + a & b_{34} \\
        0 & 0 & 0 & a_{21} & 0 & 0 & b_{43} & - b + a
    \end{bmatrix},
\eal \eeq
\beq \bal
    X_{22} = 
    \begin{bmatrix}
        b - a & b_{12} & 0 & 0 & a_{34} & 0 & 0 & 0 \\
        b_{21} & b - a & 0 & 0 & 0 & a_{34} & 0 & 0 \\
        0 & 0 & -b - a & b_{34} & 0 & 0 & a_{34} & 0 \\
        0 & 0 & b_{43} & -b - a & 0 & 0 & 0 & a_{34} \\
        a_{43} & 0 & 0 & 0 & b - a & b_{12} & 0 & 0 \\
        0 & a_{43} & 0 & 0 & b_{21} & b - a & 0 & 0 \\
        0 & 0 & a_{43} & 0 & 0 & 0 & - b - a & b_{34} \\
        0 & 0 & 0 & a_{43} & 0 & 0 & b_{43} & - b - a
    \end{bmatrix}.
\eal \eeq
Since singular values are unitarily invariant, we can perform a permutation equivalence
\beq \bal
    Y = P X P^T,
\eal \eeq
where
\beq \bal
    P = 
    \begin{bmatrix}
        I_2 & 0 & 0 & 0 & 0 & 0 & 0 & 0 \\
        0 & 0 & I_2 & 0 & 0 & 0 & 0 & 0 \\
        0 & I_2 & 0 & 0 & 0 & 0 & 0 & 0 \\
        0 & 0 & 0 & I_2 & 0 & 0 & 0 & 0 \\
        0 & 0 & 0 & 0 & I_2 & 0 & 0 & 0 \\
        0 & 0 & 0 & 0 & 0 & 0 & I_2 & 0 \\
        0 & 0 & 0 & 0 & 0 & I_2 & 0 & 0 \\
        0 & 0 & 0 & 0 & 0 & 0 & 0 & I_2
    \end{bmatrix}.
\eal \eeq
We have
\beq \bal
    Y = 
    \bma
        Y_{11} & 0 & 0 & 0 \\
        0 & Y_{22} & 0 & 0 \\
        0 & 0 & Y_{33} & 0 \\
        0 & 0 & 0 & Y_{44}
    \ema,
\eal \eeq
where
\beq \bal
    Y_{11} = 
    \bma
        b+a & b_{12} & a_{12} & 0 \\
        b_{21} & b+a & 0 & a_{12} \\
        a_{21} & 0 & b+a & b_{12} \\
        0 & a_{21} & b_{21} & b+a
    \ema,
\eal \eeq
\beq \bal
    Y_{22} =
    \bma
        -b+a & b_{34} & a_{12} & 0 \\
        b_{43} & -b+a & 0 & a_{12} \\
        a_{21} & 0 & -b+a & b_{34} \\
        0 & a_{21} & b_{43} & -b+a 
    \ema,
\eal \eeq
\beq \bal
    Y_{33} = 
    \bma
        b-a & b_{12} & a_{34} & 0 \\
        b_{21} & b-a & 0 & a_{34} \\
        a_{43} & 0 & b-a & b_{12} \\
        0 & a_{43} & b_{21} & b-a
    \ema,
\eal \eeq
\beq \bal
    Y_{44} =
    \bma
        -b-a & b_{34} & a_{34} & 0 \\
        b_{43} & -b-a & 0 & a_{34} \\
        a_{43} & 0 & -b-a & b_{34} \\
        0 & a_{43} & b_{43} & -b-a 
    \ema.
\eal \eeq
We observe that 
\beq \bal
    Y_{11} = A_{11} \otimes I_2 + I_2 \otimes B_{11}, \\
    Y_{22} = A_{11} \otimes I_2 + I_2 \otimes B_{22}, \\
    Y_{33} = A_{22} \otimes I_2 + I_2 \otimes B_{11}, \\
    Y_{44} = A_{22} \otimes I_2 + I_2 \otimes B_{22}.
\eal \eeq
To facilitate the following analysis, the Frobenius norms of $\{Y_{ii}\}_{i \in \llbracket 1,4\rrbracket}$ are listed below
\beq \bal
    \tr (Y_{11}^* Y_{11}) &= 4|b+a|^2 + 2(|a_{12}|^2 + |a_{21}|^2 + |b_{12}|^2 + |b_{21}|^2), \\
    \tr (Y_{22}^* Y_{22}) &= 4|b-a|^2 + 2(|a_{12}|^2 + |a_{21}|^2 + |b_{34}|^2 + |b_{43}|^2), \\
    \tr (Y_{33}^* Y_{33}) &= 4|b-a|^2 + 2(|a_{34}|^2 + |a_{43}|^2 + |b_{12}|^2 + |b_{21}|^2), \\
    \tr (Y_{44}^* Y_{44}) &= 4|b+a|^2 + 2(|a_{34}|^2 + |a_{43}|^2 + |b_{34}|^2 + |b_{43}|^2).
\eal \eeq
The Frobenius norm constraint \eqref{frobenius norm condition} over the entries of $A$ and $B$ is also presented
\beq \bal\label{frobenius norm condition ver. 2}
    \tr(A^* A + B^* B) = 4(|a|^2 + |b|^2) + |a_{12}|^2 + |a_{21}|^2 + |b_{12}|^2 + |b_{21}|^2 + |a_{34}|^2 + |a_{43}|^2 + |b_{34}|^2 + |b_{43}|^2 = \frac{1}{4}.
\eal \eeq
Plus, we provide the expressions of $\{Y_{ii} Y_{ii}^*\}_{i \in \llbracket 1, 4 \rrbracket}$
\beq \bal
    Y_{11} Y_{11}^* = 
    \bma
        |b + a|^2 + |b_{12}|^2 + |a_{12}|^2 & (b + a) \bar b_{21} + (\bar b + \bar a) b_{12} & (b + a) \bar a_{21} + (\bar b + \bar a) a_{12} & b_{12} \bar a_{21} + \bar b_{21} a_{12} \\
        (\bar b + \bar a) b_{21} + (b + a) \bar b_{12} & |b + a|^2 + |b_{21}|^2 + |a_{12}|^2 & b_{21} \bar a_{21} + \bar b_{12} a_{12} & (b + a) \bar a_{21} + (\bar b + \bar a) a_{12} \\
        (\bar b + \bar a) a_{21} + (b + a) \bar a_{12} & \bar b_{21} a_{21} + b_{12} \bar a_{12} & |b + a|^2 + |b_{12}|^2 + |a_{21}|^2 & (b + a) \bar b_{21} + (\bar b + \bar a) b_{12} \\
        \bar b_{12} a_{21} + b_{21} \bar a_{12} & (\bar b + \bar a) a_{21} + (b + a) \bar a_{12} & (\bar b + \bar a) b_{21} + (b + a) \bar b_{12} & |b + a|^2 + |b_{21}|^2 + |a_{21}|^2
    \ema,
\eal \eeq
\beq \bal
    Y_{22} Y_{22}^* = 
    \bma
        |a - b|^2 + |b_{34}|^2 + |a_{12}|^2 & (a - b) \bar b_{43} + (\bar a - \bar b) b_{34} & (a - b) \bar a_{21} + (\bar a - \bar b) a_{12} & b_{34} \bar a_{21} + \bar b_{43} a_{12} \\
        (\bar a - \bar b) b_{43} + (a - b) \bar b_{34} & |a - b|^2 + |b_{43}|^2 + |a_{12}|^2 & b_{43} \bar a_{21} + \bar b_{34} a_{12} & (a - b) \bar a_{21} + (\bar a - \bar b) a_{12} \\
        (\bar a - \bar b) a_{21} + (a - b) \bar a_{12} & \bar b_{43} a_{21} + b_{34} \bar a_{12} & |a - b|^2 + |b_{34}|^2 + |a_{21}|^2 & (a - b) \bar b_{43} + (\bar a - \bar b) b_{34} \\
        \bar b_{34} a_{21} + b_{43} \bar a_{12} & (\bar a - \bar b) a_{21} + (a - b) \bar a_{12} & (\bar a - \bar b) b_{43} + (a - b) \bar b_{34} & |a - b|^2 + |b_{43}|^2 + |a_{21}|^2
    \ema,
\eal \eeq
\beq \bal
    Y_{33} Y_{33}^* = 
    \bma
        |b - a|^2 + |b_{12}|^2 + |a_{34}|^2 & (b - a) \bar b_{21} + (\bar b - \bar a) b_{12} & (b - a) \bar a_{43} + (\bar b - \bar a) a_{34} & b_{12} \bar a_{43} + \bar b_{21} a_{34} \\
        (\bar b - \bar a) b_{21} + (b - a) \bar b_{12} & |b - a|^2 + |b_{21}|^2 + |a_{34}|^2 & b_{21} \bar a_{43} + \bar b_{12} a_{34} & (b - a) \bar a_{43} + (\bar b - \bar a) a_{34} \\
        (\bar b - \bar a) a_{43} + (b - a) \bar a_{34} & \bar b_{21} a_{43} + b_{12} \bar a_{34} & |b - a|^2 + |b_{12}|^2 + |a_{43}|^2 & (b - a) \bar b_{21} + (\bar b - \bar a) b_{12} \\
        \bar b_{12} a_{43} + b_{21} \bar a_{34} & (\bar b - \bar a) a_{43} + (b - a) \bar a_{34} & (\bar b - \bar a) b_{21} + (b - a) \bar b_{12} & |b - a|^2 + |b_{21}|^2 + |a_{43}|^2
    \ema,
\eal \eeq
\beq \bal
    Y_{44} Y_{44}^* = 
    \left[
        \begin{array}{cccc}
            |b + a|^2 + |b_{34}|^2 + |a_{34}|^2 & -(b + a) \bar b_{43} - (\bar b + \bar a) b_{34} & -(b + a) \bar a_{43} - (\bar b + \bar a) a_{34} & b_{34} \bar a_{43} + \bar b_{43} a_{34} \\
            -(\bar b + \bar a) b_{43} - (b + a) \bar b_{34} & |b + a|^2 + |b_{43}|^2 + |a_{34}|^2 & b_{43} \bar a_{43} + \bar b_{34} a_{34} & -(b + a) \bar a_{43} - (\bar b + \bar a) a_{34} \\
            -(\bar b + \bar a) a_{43} - (b + a) \bar a_{34} & \bar b_{43} a_{43} + b_{34} \bar a_{34} & |b + a|^2 + |b_{34}|^2 + |a_{43}|^2 & -(b + a) \bar b_{43} - (\bar b + \bar a) b_{34} \\
            \bar b_{34} a_{43} + b_{43} \bar a_{34} & -(\bar b + \bar a) a_{43} - (b + a) \bar a_{34} & -(\bar b + \bar a) b_{43} - (b + a) \bar b_{34} & |b + a|^2 + |b_{43}|^2 + |a_{43}|^2
        \end{array}
    \right].
\eal \eeq
There are three different cases we have to consider
\begin{enumerate}[label = (\roman*)]
    \item The two largest singular values are located in one block of $Y$. Without loss of generality, we can assume that these two singular values are the two largest singular values of $Y_{11}$, \label{1st case: 2 2 bloc-diag}
    \item The two largest singular values are found in two different blocks of $Y$, which are Kronecker sums of one block from $A$ (or respectively $B$) and two blocks from $B$ (or respectively $A$). We can assume that $\sigma_1(Y)$ is the largest singular value of $Y_{11}$, and $\sigma_2(Y)$ is the largest singular value of $Y_{22}$, \label{2nd case: 2 2 bloc-diag}
    \item The two largest singular values are located respectively in two different blocks of $Y$, which are Kronecker sums of two blocks from $A$ and two blocks from $B$. We can assume that $\sigma_1(Y)$ is the largest singular value of $Y_{11}$, and $\sigma_2(Y)$ is the largest singular value of $Y_{44}$. \label{3rd case: 2 2 bloc-diag}
\end{enumerate}

In the following passage, we will perform a case-by-case analysis to demonstrate the validity of distillability conjecture in all 3 cases.

\vspace{7pt}
\textbf{Initially, let us examine case \ref{1st case: 2 2 bloc-diag}.}

\begin{lemma}\label{lemma: 1st case 2 2 bloc-diag}
    In case \ref{1st case: 2 2 bloc-diag}, $\sigma_1(Y)^2 + \sigma_2(Y)^2 \leq \frac{1}{2}$.
\end{lemma}

\begin{proof}
    Suppose that $\sigma_1(Y_{11})^2 + \sigma_2(Y_{11})^2$ achieve its maximum when $\norm{A_{11}}_F^2 + \norm{B_{11}}_F^2 + 2(|a|^2 + |b|^2) = c_1 > 0$, and $|a_{34}|^2 + |a_{43}|^2 + |b_{34}|^2 + |b_{43}|^2 = c_2 > 0$. By scaling the entries of $A_{11}$ and $B_{11}$ by a factor of $\sqrt{\frac{c_2 + c_1}{c_1}}$ and setting $a_{34} = a_{43} = b_{34} = b_{43} = 0$, i.e. multiplying $Y_{11}$ by $\sqrt{\frac{c_2 + c_1}{c_1}}$, the singular values of $Y_{11}$ are multiplied by $\sqrt{\frac{c_2 + c_1}{c_1}}$ while preserving the Frobenius norm constraint \eqref{frobenius norm condition}. This contradicts the premise that $\sigma_1(Y_{11})^2 + \sigma_2(Y_{11})^2$ is maximized when $|a_{34}|^2 + |a_{43}|^2 + |b_{34}|^2 + |b_{43}|^2 = c_2 > 0$. Therefore, $\sigma_1(Y_{11})^2 + \sigma_2(Y_{11})^2$ is maximized when $a_{34} = a_{43} = b_{34} = b_{43} = 0$.
    \beq \bal
        \sigma_1(Y)^2 + \sigma_2(Y)^2 &= \sigma_1(Y_{11})^2 + \sigma_2(Y_{11})^2 \leq \norm{Y_{11}}_F^2 = \norm{Y}_F^2 - \norm{Y_{22}}_F^2 - \norm{Y_{33}}_F^2 - \norm{Y_{44}}_F^2 \\
        &= 1 - 8|a-b|^2 - 4|a+b|^2 - 2(|a_{12}|^2 + |a_{21}|^2 + |b_{12}|^2 + |b_{21}|^2) \\
        &= 1 - (4|a-b|^2 + 8|a|^2 + 8|b|^2 + 2(|a_{12}|^2 + |a_{21}|^2 + |b_{12}|^2 + |b_{21}|^2)).
    \eal \eeq
    Recall that $\norm{A}_F^2 + \norm{B}_F^2 = 4|a|^2 + 4|b|^2 + |a_{12}|^2 + |a_{21}|^2 + |b_{12}|^2 + |b_{21}|^2 = \frac{1}{4}$, we have
    \beq \bal
        \sigma_1(Y)^2 + \sigma_2(Y)^2 &= \frac{1}{2} - 4|a-b|^2 \leq \frac{1}{2}.
    \eal \eeq
\end{proof}

\textbf{Subsequently, we examine case \ref{2nd case: 2 2 bloc-diag}.}

\begin{lemma}\label{lemma: 2nd case 2 2 bloc-diag}
    In case \ref{2nd case: 2 2 bloc-diag}, $\sigma_1(Y)^2 + \sigma_2(Y)^2 \leq \frac{1}{2}$.
\end{lemma}
\begin{proof}
    Suppose that $\sigma_1(Y_{11})^2 + \sigma_1(Y_{22})^2$ achieve its maximum when $4|a_{34}|^2 + 4|a_{43}|^2 = c_1 > 0, \norm{Y_{11}}_F^2 + \norm{Y_{22}}_F^2 = c_2 > 0$, and $\norm{Y_{33}}_F^2 + \norm{Y_{44}}_F^2 - 4|a_{34}|^2 - 4|a_{43}|^2 = c_3 > 0$. We can easily verify that $c_1 + c_2 + c_3 = 1$. By scaling the entries of $A_{11}$, $B_{11}$, and $B_{22}$ by a factor of $\sqrt{\frac{1}{c_2 + c_3}}$ and setting $a_{34} = a_{43} = 0$, i.e. muliplying $Y_{11}$, $Y_{22}$ by $\sqrt{\frac{1}{c_2 + c_3}}$, the singular values of $Y_{11}$ and $Y_{22}$ are multiplied by $\sqrt{\frac{1}{c_2 + c_3}} > 1$, while preserving the Frobenius norm constraint \eqref{frobenius norm condition}. This contradicts the premise that $\sigma_1(Y_{11})^2 + \sigma_1(Y_{22})^2$ is maximized when $4|a_{34}|^2 + 4|a_{43}|^2 = c_1 > 0$. Therefore, $\sigma_1(Y_{11})^2 + \sigma_1(Y_{22})^2$ is maximized when $a_{34} = a_{43} = 0$.

    \vspace{7pt}
    Consider $Y_{11} Y_{11}^*$ and $Y_{22} Y_{22}^*$
    \beq \bal
        Z_{11} = Y_{11} {Y_{11}}^* = 
        \begin{bmatrix}
            {Z_{11}}^{(11)} & {Z_{11}}^{(12)} \\
            {{Z_{11}}^{(12)}}^* & {Z_{11}}^{(22)}
        \end{bmatrix},
        \quad
        Z_{22} = Y_{22} {Y_{22}}^* = 
        \begin{bmatrix}
            {Z_{22}}^{(11)} & {Z_{22}}^{(12)} \\
            {{Z_{22}}^{(12)}}^* & {Z_{22}}^{(22)}
        \end{bmatrix},
    \eal \eeq
    where
    \beq \bal
        {Z_{11}}^{(11)} = 
        \bma
            |b + a|^2 + |b_{12}|^2 + |a_{12}|^2 & (b + a) \bar b_{21} + (\bar b + \bar a) b_{12} \\
            (\bar b + \bar a) b_{21} + (b + a) \bar b_{12} & |b + a|^2 + |b_{21}|^2 + |a_{12}|^2
        \ema, \\
        {Z_{11}}^{(22)} = 
        \bma
            |b + a|^2 + |b_{12}|^2 + |a_{21}|^2 & (b + a) \bar b_{21} + (\bar b + \bar a) b_{12} \\
            (\bar b + \bar a) b_{21} + (b + a) \bar b_{12} & |b + a|^2 + |b_{21}|^2 + |a_{21}|^2
        \ema, \\
        {Z_{22}}^{(11)} = 
        \bma
            |a - b|^2 + |b_{34}|^2 + |a_{12}|^2 & (a - b) \bar b_{43} + (\bar a - \bar b) b_{34} \\
            (\bar a - \bar b) b_{43} + (a - b) \bar b_{34} & |a - b|^2 + |b_{43}|^2 + |a_{12}|^2 
        \ema, \\
        {Z_{22}}^{(22)} = 
        \bma
            |a - b|^2 + |b_{34}|^2 + |a_{21}|^2 & (a - b) \bar b_{43} + (\bar a - \bar b) b_{34} \\
            (\bar a - \bar b) b_{43} + (a - b) \bar b_{34} & |a - b|^2 + |b_{43}|^2 + |a_{21}|^2
        \ema.
    \eal \eeq
    The largest eigenvalue of a positive semidefinite matrix is bounded above by the sum of the largest eigenvalues of two diagonal blocks of the positive semidefinite matrix. For more details, please refer to \cite{BOURIN20121906}.
    \begin{equation}
    \begin{aligned}
        \sigma_1(Y_{11})^2 +& \sigma_1(Y_{22})^2 \\
        =& \lambda_{\max}(Z_{11}) + \lambda_{\max}(Z_{22}) \\ 
        \leq& \lambda_{\max}({Z_{11}}^{(11)}) + \lambda_{\max}({Z_{11}}^{(22)}) + \lambda_{\max}({Z_{22}}^{(11)}) + \lambda_{\max}({Z_{22}}^{(22)}) \\
        =& \frac{1}{2} \left(2 \left|a_{12}\right|^2 + 2 \left|a + b\right|^2 + \left|b_{12}\right|^2 + \left|b_{21}\right|^2 + \sqrt{\left(\left|b_{12}\right|^2 - \left|b_{21}\right|^2\right)^2 + 4\left|b_{21} \left(\bar a + \bar b\right) + \bar b_{12} \left(a + b\right)\right|^2}\right) + \\
        & \frac{1}{2} \left(2 \left|a_{21}\right|^2 + 2 \left|a + b\right|^2 + \left|b_{12}\right|^2 + \left|b_{21}\right|^2 + \sqrt{\left(\left|b_{12}\right|^2 - \left|b_{21}\right|^2\right)^2 + 4\left|b_{21} \left(\bar a + \bar b\right) + \bar b_{12} \left(a + b\right)\right|^2}\right) + \\
        & \frac{1}{2} \left(2 \left|a_{12}\right|^2 + 2 \left|a - b\right|^2 + \left|b_{34}\right|^2 + \left|b_{43}\right|^2 + \sqrt{\left(\left|b_{34}\right|^2 - \left|b_{43}\right|^2\right)^2 + 4\left|b_{43} \left(\bar a - \bar b\right) + \bar b_{34} \left(a - b\right)\right|^2}\right) + \\
        & \frac{1}{2} \left(2 \left|a_{21}\right|^2 + 2 \left|a - b\right|^2 + \left|b_{34}\right|^2 + \left|b_{43}\right|^2 + \sqrt{\left(\left|b_{34}\right|^2 - \left|b_{43}\right|^2\right)^2 + 4\left|b_{43} \left(\bar a - \bar b\right) + \bar b_{34} \left(a - b\right)\right|^2}\right) \\
        =& 8|a|^2 + 8|b|^2 + 2(|a_{12}|^2 + |a_{21}|^2 + |b_{12}|^2 + |b_{21}|^2 + |b_{34}|^2 + |b_{43}|^2) = \frac{1}{2}.
    \end{aligned}
    \end{equation}
    For more details of the proof of this inequality, please refer to section \ref{sec: proof of inequality} in the appendix.
\end{proof}

\vspace{7pt}
\textbf{At last, we examine the third case \ref{3rd case: 2 2 bloc-diag}.}

\begin{lemma}\label{lemma: 3rd case 2 2 bloc-diag}
    In case \ref{3rd case: 2 2 bloc-diag}, $\sigma_1(Y)^2 + \sigma_2(Y)^2 \leq \frac{1}{2}$.
\end{lemma}
\begin{proof}
    Suppose that $\sigma_1(Y_{11})^2 + \sigma_1(Y_{44})^2$ achieve its maximum when $8|b - a|^2 = c_1 > 0, \norm{Y_{11}}_F^2 + \norm{Y_{44}}_F^2 = c_2 > 0$, and $2(|a_{12}|^2 + |a_{21}|^2 + |b_{34}|^2 + |b_{43}|^2 + |a_{34}|^2 + |a_{43}|^2 + |b_{12}|^2 + |b_{21}|^2) = c_3 > 0$. We can easily verify that $c_1 + c_2 + c_3 = 1$. By scaling the entries of $Y_{11}$ and $Y_{44}$ by a factor of $\sqrt{\frac{1}{c_2 + c_3}}$ and setting $a = b$, the singular values of $Y_{11}$ and $Y_{44}$ are multiplied by $\sqrt{\frac{1}{c_2 + c_3}} > 1$, while preserving the Frobenius norm constraint \eqref{frobenius norm condition}. This contradicts the premise that $\sigma_1(Y_{11})^2 + \sigma_1(Y_{44})^2$ is maximized when $8|b - a|^2 = c_1 > 0$. Therefore, $\sigma_1(Y_{11})^2 + \sigma_1(Y_{44})^2$ is maximized when $a = b$.

\vspace{7pt}
Besides, since 
\beq \bal
    &\sigma_1(Y_{11}) = \max_{\|x\|_F = 1}\|Y_{11} x\|_F \leq \max_{\|x\|_F = 1}\||Y_{11}| |x|\|_F \leq \sigma_1(|Y_{11}|), \\
    &\sigma_1(Y_{44}) = \max_{\|x\|_F = 1}\|Y_{22} x\|_F \leq \max_{\|x\|_F = 1}\|(-|Y_{44}|) (-|x|)\|_F \leq \sigma_1(-|Y_{44}|),
\eal \eeq
where $|\cdot|$ denotes the entry-wise absolute value of a matrix or a vector. We obtain that $\sigma_1(Y_{11})^2 + \sigma_1(Y_{44})^2$ is maximized when all the entries of $Y_{11}$ are nonnegative, and all the entries of $Y_{44}$ are nonpositive.

\vspace{7pt}
Based on the conclusions above, we can simplify the expression of $Y_{11}$ and $Y_{44}$ as follows
\beq \bal
    Y_{11} = 
    \begin{bmatrix}
        2a & b_{12} & a_{12} & 0 \\
        b_{21} & 2a & 0 & a_{12} \\
        a_{21} & 0 & 2a & b_{12} \\
        0 & a_{21} & b_{21} & 2a
    \end{bmatrix}, \quad
    Y_{44} = 
    \begin{bmatrix}
        -2a & b_{34} & a_{34} & 0 \\
        b_{43} & -2a & 0 & a_{34} \\
        a_{43} & 0 & -2a & b_{34} \\
        0 & a_{43} & b_{43} & -2a
    \end{bmatrix}.
\eal \eeq
    $Y_{11}$ and $Y_{44}$ can be decomposed into the following form
    \begin{align}
        &Y_{11} = {Y_{11}}^{(1)} + {Y_{11}}^{(2)} = 
        \begin{bmatrix}
            2a & 0 & 0 & 0 \\
            0 & 2a & 0 & 0 \\
            0 & 0 & 2a & 0 \\
            0 & 0 & 0 & 2a
        \end{bmatrix}
        +
        \begin{bmatrix}
            0 & b_{12} & a_{12} & 0 \\
            b_{21} & 0 & 0 & a_{12} \\
            a_{21} & 0 & 0 & b_{12} \\
            0 & a_{21} & b_{21} & 0
        \end{bmatrix}, \\
        &Y_{44} = {Y_{44}}^{(1)} + {Y_{44}}^{(2)} = 
        \begin{bmatrix}
            -2a & 0 & 0 & 0 \\
            0 & -2a & 0 & 0 \\
            0 & 0 & -2a & 0 \\
            0 & 0 & 0 & -2a
        \end{bmatrix}
        +
        \begin{bmatrix}
            0 & b_{34} & a_{34} & 0 \\
            b_{43} & 0 & 0 & a_{34} \\
            a_{43} & 0 & 0 & b_{34} \\
            0 & a_{43} & b_{43} & 0
        \end{bmatrix}.
    \end{align}
    Then, according to the triangular inequality of the operator norm, which is the largest singular value in our scenario, we have
    \beq \bal
        \sigma_1(Y)^2 + \sigma_2(Y)^2 &= \sigma_1(Y_{11})^2 + \sigma_1(Y_{44})^2 \leq \left(\sigma_1({Y_{11}}^{(1)}) + \sigma_1({Y_{11}}^{(2)})\right)^2 + \left(\sigma_1({Y_{44}}^{(1)}) + \sigma_1({Y_{44}}^{(2)})\right)^2 \\
        &\leq 2 \left(\sigma_1({Y_{11}}^{(1)})^2 + \sigma_1({Y_{11}}^{(2)})^2 + \sigma_1({Y_{44}}^{(1)})^2 + \sigma_1({Y_{44}}^{(2)})^2\right).
    \eal \eeq
    Without loss of generality, we can only work with $Y_{11}$ and apply the result to $Y_{44}$
    \beq \bal
        \sigma_1({Y_{11}}^{(2)})^2 &= \lambda_{\max}({Y_{11}}^{(2)} {{Y_{11}}^{(2)}}^*) \\
        &= \frac{1}{2} \left(a_{12}^2 + a_{21}^2 + b_{12}^2 + b_{21}^2 + \sqrt{\left(\left(a_{12} + a_{21}\right)^2 + \left(b_{12} - b_{21}\right)^2\right) \left(\left(a_{12} - a_{21}\right)^2 + \left(b_{12} + b_{21}\right)^2\right)}\right) \\
        &\leq \frac{1}{2} \left(a_{12}^2 + a_{21}^2 + b_{12}^2 + b_{21}^2 + \frac{1}{2} \left(\left(a_{12} + a_{21}\right)^2 + \left(b_{12} - b_{21}\right)^2 + \left(a_{12} - a_{21}\right)^2 + \left(b_{12} + b_{21}\right)^2\right)\right) \\
        &= a_{12}^2 + a_{21}^2 + b_{12}^2 + b_{21}^2.
    \eal \eeq
    In the same way, we obtain
    \beq \bal
        \sigma_1({Y_{44}}^{(2)})^2 \leq a_{34}^2 + a_{43}^2 + b_{34}^2 + b_{43}^2.
    \eal \eeq
    Now, we have
    \beq \bal
        \sigma_1(Y)^2 + \sigma_2(Y)^2 &\leq 2 \left(\sigma_1({Y_{11}}^{(1)})^2 + \sigma_1({Y_{11}}^{(2)})^2 + \sigma_1({Y_{44}}^{(1)})^2 + \sigma_1({Y_{44}}^{(2)})^2\right) \\
        &= 2(8a^2 + a_{12}^2 + a_{21}^2 + b_{12}^2 + b_{21}^2 + a_{34}^2 + a_{43}^2 + b_{34}^2 + b_{43}^2).
    \eal \eeq
    Recall that $\norm{A}_F^2 + \norm{B}_F^2 = \frac{1}{4} = 8a^2 + a_{12}^2 + a_{21}^2 + b_{12}^2 + b_{21}^2 + a_{34}^2 + a_{43}^2 + b_{34}^2 + b_{43}^2$, we have $\sigma_1(Y)^2 + \sigma_2(Y)^2 \leq \frac{1}{2}$.
\end{proof}

\vspace{7pt}
\textbf{According to Lemmas \ref{lemma: 1st case 2 2 bloc-diag}, \ref{lemma: 2nd case 2 2 bloc-diag}, and \ref{lemma: 3rd case 2 2 bloc-diag}, cases \ref{1st case: 2 2 bloc-diag}, \ref{2nd case: 2 2 bloc-diag}, and \ref{3rd case: 2 2 bloc-diag} have all been proven. Therefore, Theorem \ref{thm: 2 2 bloc-diag} holds.}

\subsection{$B$ is unitarily equivalent to $-A^T$ or $-A$}
\label{sec:UE=-ATor-A}

The main conclusion of this subsection is Theorem \ref{thm: B = -A^T or -A}, which is supported by the following analysis. First, we reformulate Conjecture \ref{conjecture: distillability} into another equivalent form \eqref{eq form}. Then, building upon the foundational works on the Frobenius norm of the commutator and its variants, we establish the proof of Theorem \ref{thm: B = -A^T or -A}. 

\begin{theorem}\label{thm: B = -A^T or -A}
    Conjecture \ref{conjecture: distillability} holds when $B$ is unitarily equivalent to either $-A^T$ or $-A$.
\end{theorem}

The proof of Theorem \ref{thm: B = -A^T or -A} is supported by the following analysis, Lemmas \ref{lemma: B = -A^T} and \ref{lemma: B = -A}.

\vspace{7pt}
It is known that the sum of squares of the two largest singular values of $X$ is equal to the sum of the two largest eigenvalues of $X^* X$. According to the variational characterization of eigenvalues of Hermitian matrices, we have
\beq \label{equality 1}
\bal
    \s_1(X)^2 + \s_2(X)^2 &= \l_1(X^* X)^2 + \l_2(X^* X)^2 \\
    &= \max_{\substack{U \in M_{16,2}(\bbC) \\ U^* U = I_2}} \tr U^* X^* X U.
\eal
\eeq
For more details, please refer to Corollary 4.3.39. in \cite{Horn_Johnson_2012}. 

\vspace{7pt}
We partition $U = \bma u_1 & u_2 \ema$ with $u_1, u_2 \in \bbC^{16}$
\beq \label{equality 2}
\bal
    \max_{\substack{U \in M_{16,2} \\ U^* U = I_2}} \tr U^* X^* X U &= \max_{\substack{u_1, u_2 \in \bbC^{16}, \braket{u_1}{u_2} = 0 \\ \norm{u_1}_F = \norm{u_2}_F = 1}} \tr \bma
        u_1^* X^* X u_1 & u_1^* X^* X u_2 \\
        u_2^* X^* X u_1 & u_2^* X^* X u_2
    \ema \\
    &= \max_{\substack{u_1, u_2 \in \bbC^{16}, \braket{u_1}{u_2} = 0 \\ \norm{u_1}_F = \norm{u_2}_F = 1}} u_1^* X^* X u_1 + u_2^* X^* X u_2 \\
    & = \max_{\substack{u_1, u_2 \in \bbC^{16}, \braket{u_1}{u_2} = 0 \\ \norm{u_1}_F = \norm{u_2}_F = 1}} \norm{X u_1}_F^2 + \norm{X u_2}_F^2. 
\eal 
\eeq

Consider the equalities \eqref{equality 1} and \eqref{equality 2} we proven above
\beq \bal
    \s_1(X)^2 + \s_2(X)^2 = \max_{\substack{\operatorname{vec} V_1, \operatorname{vec} V_2 \in \bbC^{16} \\ \braket{\operatorname{vec} V_1}{\operatorname{vec} V_2} = 0 \\ \norm{\operatorname{vec} V_1}_F = \norm{\operatorname{vec} V_2}_F = 1}} \norm{X \operatorname{vec} V_1}_F^2 + \norm{X \operatorname{vec} V_2}_F^2,
\eal \eeq
where $V_1, V_2 \in M_4(\bbC)$ (for the definition of $\operatorname{vec}$, please refer to Definition \ref{def: vec}). Thanks to Lemma \ref{lemma: krom sum and commutators}, we have
\beq \label{eq form} \bal
    \s_1(X)^2 + \s_2(X)^2 = \max_{\substack{V_1, V_2 \in \bbC^{16}, \tr V_1^* V_2 = 0 \\ \norm{V_1}_F = \norm{V_2}_F = 1}} \norm{B V_1 + V_1 A^T}_F^2 + \norm{B V_2 + V_2 A^T}_F^2.
\eal \eeq

\textbf{We first demonstrate that Conjecture \ref{conjecture: distillability} is valid when $B$ is unitarily equivalent to $-A^T$.}

\begin{lemma}\label{lemma: B = -A^T}
    Under the context of Conjecture \ref{conjecture: distillability}, $\s_1(X)^2 + \s_2(X)^2 \leq \frac{1}{2}$ when $B = - U A^T U^*$, where $U$ is a $4$ by $4$ unitary matrix. 
\end{lemma}

\bpf
    For $B = - U A^T U^*$, we have
    \beq \bal
        \s_1(X)^2 + \s_2(X)^2 &= \max_{\substack{V_1, V_2 \in \bbC^{16}, \tr V_1^* V_2 = 0 \\ \norm{V_1}_F = \norm{V_2}_F = 1}} \norm{- U A^T U^* V_1 + V_1 A^T}_F^2 + \norm{- U A^T U^* V_2 + V_2 A^T}_F^2 \\
        &= \max_{\substack{V_1, V_2 \in \bbC^{16}, \tr V_1^* V_2 = 0 \\ \norm{V_1}_F = \norm{V_2}_F = 1}} \norm{- A^T U^* V_1 + U^* V_1 A^T}_F^2 + \norm{- A^T U^* V_2 + U^* V_2 A^T}_F^2.
    \eal \eeq
    According to \eqref{xy-yx} of Lemma \ref{lemma: frob norm xy-yx}
    \beq \bal
        \forall Y, Z \in M_n(\bbC), \quad \norm{YZ - ZY}_F \leq \sqrt 2 \norm{Y}_F \norm{Z}_F,
    \eal \eeq
    we have
    \beq \bal
        \s_1(X)^2 + \s_2(X)^2 \leq 2 (\norm{A^T}_F^2 \norm{U^* V_1}_F^2 + \norm{A^T}_F^2 \norm{U^* V_2}_F^2) = \frac{1}{2}.
    \eal \eeq
\epf

\textbf{Similarly, we can prove that Conjecture \ref{conjecture: distillability} holds when $B$ is unitarily equivalent to $-A$.}

\begin{lemma}\label{lemma: B = -A}
    Under the context of Conjecture \ref{conjecture: distillability}, $\s_1(X)^2 + \s_2(X)^2 \leq \frac{1}{2}$ when $B = - U A U^*$, where $U$ is a $4$ by $4$ unitary matrix. 
\end{lemma}

\bpf
    For $B = - U A U^*$, we have
    \beq \bal
        \s_1(X)^2 + \s_2(X)^2 &= \max_{\substack{V_1, V_2 \in \bbC^{16}, \tr V_1^* V_2 = 0 \\ \norm{V_1}_F = \norm{V_2}_F = 1}} \norm{- U A U^* V_1 + V_1 A^T}_F^2 + \norm{- U A U^* V_2 + V_2 A^T}_F^2 \\
        &= \max_{\substack{V_1, V_2 \in \bbC^{16}, \tr V_1^* V_2 = 0 \\ \norm{V_1}_F = \norm{V_2}_F = 1}} \norm{- A U^* V_1 + U^* V_1 A^T}_F^2 + \norm{- A U^* V_2 + U^* V_2 A^T}_F^2.
    \eal \eeq
    According to \eqref{xy-yxt} of Lemma \ref{lemma: frob norm xy-yx}
    \beq \bal
        \forall Y, Z \in M_n(\bbC), \quad \norm{YZ - ZY^T}_F \leq \sqrt 2 \norm{Y}_F \norm{Z}_F,
    \eal \eeq
    we have
    \beq \bal
        \s_1(X)^2 + \s_2(X)^2 \leq 2 (\norm{A^T}_F^2 \norm{U^* V_1}_F^2 + \norm{A^T}_F^2 \norm{U^* V_2}_F^2) = \frac{1}{2}.
    \eal \eeq
\epf

\textbf{According to Lemmas \ref{lemma: B = -A^T} and \ref{lemma: B = -A}, Theorem \ref{thm: B = -A^T or -A} holds.}

\subsection{Matrices with distinct eigenvalues} \label{subsec: diag}

In this subsection, a simplified version of the distillability conjecture is introduced and presented as Theorem \ref{thm: diag prob set}. Its proof is established through Lemmas \ref{lemma: phi} and \ref{lemma: dense}. Detailed discussions on continuity, compactness, and density are also provided.

\begin{theorem} \label{thm: diag prob set}
    Conjecture \ref{conjecture: distillability} holds for the set $\cX$ defined as follows
    \begin{align}
        \left\{A, B \in M_4(\bbC)\left| \tr A = \tr B = 0, \norm{A}_F^2 + \norm{B}_F^2 = \frac{1}{4}\right\}\right.,
    \end{align}
    if and only if Conjecture \ref{conjecture: distillability} holds for the set $\cY$ defined as follows
    \begin{equation}
        \left\{
        \begin{array}{c|c}
            \multirow{2}{*}{$A, B \in M_4(\bbC)$} & \tr A = \tr B = 0, \norm{A}_F^2 + \norm{B}_F^2 = \frac{1}{4}, \\
            & A \text{ and } B \text{ only possess distinct eigenvalues}
        \end{array}
        \right\}.
    \end{equation}
\end{theorem}

\begin{lemma} \label{lemma: phi}
    Under the context of Conjecture \ref{conjecture: distillability}, the map $\phi$ defined as follows is continuous
    \begin{equation}
        \begin{aligned}
            \phi: M_4(\bbC) \times M_4(\bbC) &\rightarrow \bbR \\
            A, B &\mapsto \s_1(X)^2 + \s_2(X)^2.
        \end{aligned}
    \end{equation}
\end{lemma}

\begin{proof}
    According to Corollary 7.3.5 of \cite{Horn_Johnson_2012}, we have
    \begin{equation}
        \forall Y, Z \in M_n(\bbC)\quad\left|\s_i(Y) - \s_i(Z)\right| \leq \norm{Y - Z},
    \end{equation}
    where $\norm{\cdot}$ is the operator norm. Since $\norm{\cdot}$ and $\norm{\cdot}_F$ are equivalent on $M_{16}(\bbC)$, which means
    \begin{equation}
        \exists c > 0 \text{ such that } \left|\s_i(Y) - \s_i(Z)\right| \leq c\norm{Y - Z}_F,
    \end{equation}
    we obtain that the map $\varphi$ defined as follows is continuous with respect to $\norm{\cdot}_F$
    \begin{equation}
        \begin{aligned}
            \varphi: M_{16}(\bbC) &\rightarrow \bbR \\
            X &\mapsto \s_1(X)^2 + \s_2(X)^2.
        \end{aligned}
    \end{equation}
    Plus, the map $\psi: A, B \mapsto A \ox I + I \ox B = X$ is obviously continuous with respect to $\norm{\cdot}_2$ defined in Definition \ref{def: 2 norm}. Therefore, $\phi = \varphi \circ \psi$ is continuous.
\end{proof}

\begin{lemma} \label{lemma: dense}
    The set $\cY$ is a dense subset of $\cX$, with respect to the norm $\norm{\cdot}_2$ defined in Definition \ref{def: 2 norm}.
\end{lemma}

\begin{proof}
    For more details, please refer to section \ref{sec: dense proof} in the appendix.
\end{proof}

\vspace{7pt}
\textbf{Now we prove Theorem \ref{thm: diag prob set}}. Obviously, if Conjecture \ref{conjecture: distillability} holds for $\cX$, it holds for $\cY$. In the following paragraph, we prove the opposite direction.

\vspace{7pt}
The functions $\norm{\cdot}_2: (A, B) \mapsto \sqrt{\norm{A}_F^2 + \norm{B}_F^2}$ and $\tr : A \mapsto \tr A$ are continuous, and the sets $\{\frac{1}{4}\}$ and $\{0\}$ are closed. The set $\cX$ is defined as the intersection of the preimages of ${\frac{1}{4}}$ and ${0}$ under $\norm{\cdot}_2$ and $\tr$, respectively. Since the preimage of a closed set under a continuous function is closed, $\cX$ is closed in a finite-dimensional space and is therefore compact.

\vspace{7pt}
Furthermore, the map $\phi$ defined in Lemma \ref{lemma: phi} is continuous, meaning $\phi$ attains its maximum on the compact set $\cX$. If there exists a pair $(A_x, B_x) \in \cX$ such that $\phi(A_x, B_x) = p > \frac{1}{2}$, then, according to Lemma \ref{lemma: dense}, a matrix pair $(A_y, B_y) \in \cY$ can be chosen to be arbitrarily close to $(A_x, B_x)$. Consequently, $\phi(A_y, B_y)$ can be made arbitrarily close to $\phi(A_x, B_x)$, surpassing $\frac{1}{2}$.

\vspace{7pt}
By contrapositive reasoning, if Conjecture \ref{conjecture: distillability} holds for $\cX$, it also holds for $\cY$. \textbf{At this point, Theorem \ref{thm: diag prob set} is established.}

\section{Conclusions}
\label{sec:con}

We have investigated the long-standing open problem on two-copy Werner states' distillability conjecture in terms of $4$ by $4$ matrices. We established and proved two specific instances of the conjecture. Our results represent the latest progress on the conjecture. The next step is to study more cases such as the direct sum of $3\times3$ matrices and a number. Another approach is to explore more mathematical properties of the Frobenius norm of commutators employed in this paper. We could also exploit the connection between matrices with distinct eigenvalues and the distillability conjecture, simplifying arguments by using diagonalizable matrices.

\section*{ACKNOWLEDGMENTS}
Authors were supported by the NNSF of China (Grant No. 12471427).

\appendix
\section{Proof of inequality in Lemma \ref{lemma: 2nd case 2 2 bloc-diag}}\label{sec: proof of inequality}
\begin{equation}
\begin{aligned}
        \sigma_1(Y_{11})^2 +& \sigma_1(Y_{22})^2 \\
        =& \lambda_{\max}(Z_{11}) + \lambda_{\max}(Z_{22}) \\ 
        \leq& \lambda_{\max}({Z_{11}}^{(11)}) + \lambda_{\max}({Z_{11}}^{(22)}) + \lambda_{\max}({Z_{22}}^{(11)}) + \lambda_{\max}({Z_{22}}^{(22)}) \\
        =& \frac{1}{2} \left(2 \left|a_{12}\right|^2 + 2 \left|a + b\right|^2 + \left|b_{12}\right|^2 + \left|b_{21}\right|^2 + \sqrt{\left(\left|b_{12}\right|^2 - \left|b_{21}\right|^2\right)^2 + 4\left|b_{21} \left(\bar a + \bar b\right) + \bar b_{12} \left(a + b\right)\right|^2}\right) + \\
        & \frac{1}{2} \left(2 \left|a_{21}\right|^2 + 2 \left|a + b\right|^2 + \left|b_{12}\right|^2 + \left|b_{21}\right|^2 + \sqrt{\left(\left|b_{12}\right|^2 - \left|b_{21}\right|^2\right)^2 + 4\left|b_{21} \left(\bar a + \bar b\right) + \bar b_{12} \left(a + b\right)\right|^2}\right) + \\
        & \frac{1}{2} \left(2 \left|a_{12}\right|^2 + 2 \left|a - b\right|^2 + \left|b_{34}\right|^2 + \left|b_{43}\right|^2 + \sqrt{\left(\left|b_{34}\right|^2 - \left|b_{43}\right|^2\right)^2 + 4\left|b_{43} \left(\bar a - \bar b\right) + \bar b_{34} \left(a - b\right)\right|^2}\right) + \\
        & \frac{1}{2} \left(2 \left|a_{21}\right|^2 + 2 \left|a - b\right|^2 + \left|b_{34}\right|^2 + \left|b_{43}\right|^2 + \sqrt{\left(\left|b_{34}\right|^2 - \left|b_{43}\right|^2\right)^2 + 4\left|b_{43} \left(\bar a - \bar b\right) + \bar b_{34} \left(a - b\right)\right|^2}\right) \\
        \leq& \frac{1}{2} \left(2 \left|a_{12}\right|^2 + 2 \left|a + b\right|^2 + \left|b_{12}\right|^2 + \left|b_{21}\right|^2 + \sqrt{\left(\left|b_{12}\right| - \left|b_{21}\right|\right)^2 \left(\left|b_{12}\right| + \left|b_{21}\right|\right)^2 + 4\left(\left|b_{21}\right|+\left|b_{12}\right|\right)^2\left|a + b\right|^2}\right) + \\
        & \frac{1}{2} \left(2 \left|a_{21}\right|^2 + 2 \left|a + b\right|^2 + \left|b_{12}\right|^2 + \left|b_{21}\right|^2 + \sqrt{\left(\left|b_{12}\right| - \left|b_{21}\right|\right)^2 \left(\left|b_{12}\right| + \left|b_{21}\right|\right)^2 + 4\left(\left|b_{21}\right|+\left|b_{12}\right|\right)^2\left|a + b\right|^2}\right) + \\
        & \frac{1}{2} \left(2 \left|a_{12}\right|^2 + 2 \left|a - b\right|^2 + \left|b_{34}\right|^2 + \left|b_{43}\right|^2 + \sqrt{\left(\left|b_{34}\right| - \left|b_{43}\right|\right)^2 \left(\left|b_{34}\right| + \left|b_{43}\right|\right)^2 + 4\left(\left|b_{43}\right|+\left|b_{34}\right|\right)^2\left|a - b\right|^2}\right) + \\
        & \frac{1}{2} \left(2 \left|a_{21}\right|^2 + 2 \left|a - b\right|^2 + \left|b_{34}\right|^2 + \left|b_{43}\right|^2 + \sqrt{\left(\left|b_{34}\right| - \left|b_{43}\right|\right)^2 \left(\left|b_{34}\right| + \left|b_{43}\right|\right)^2 + 4\left(\left|b_{43}\right|+\left|b_{34}\right|\right)^2\left|a - b\right|^2}\right) \\
        =& \frac{1}{2} \left(2 \left|a_{12}\right|^2 + 2 \left|a + b\right|^2 + \left|b_{12}\right|^2 + \left|b_{21}\right|^2 + \left(\left|b_{21}\right|+\left|b_{12}\right|\right) \sqrt{\left(\left|b_{12}\right| - \left|b_{21}\right|\right)^2 + 4\left|a + b\right|^2}\right) + \\
        & \frac{1}{2} \left(2 \left|a_{21}\right|^2 + 2 \left|a + b\right|^2 + \left|b_{12}\right|^2 + \left|b_{21}\right|^2 + \left(\left|b_{21}\right|+\left|b_{12}\right|\right) \sqrt{\left(\left|b_{12}\right| - \left|b_{21}\right|\right)^2 + 4\left|a + b\right|^2}\right) + \\
        & \frac{1}{2} \left(2 \left|a_{12}\right|^2 + 2 \left|a - b\right|^2 + \left|b_{34}\right|^2 + \left|b_{43}\right|^2 + \left(\left|b_{43}\right|+\left|b_{34}\right|\right) \sqrt{\left(\left|b_{34}\right| - \left|b_{43}\right|\right)^2 + 4\left|a - b\right|^2}\right) + \\
        & \frac{1}{2} \left(2 \left|a_{21}\right|^2 + 2 \left|a - b\right|^2 + \left|b_{34}\right|^2 + \left|b_{43}\right|^2 + \left(\left|b_{43}\right|+\left|b_{34}\right|\right) \sqrt{\left(\left|b_{34}\right| - \left|b_{43}\right|\right)^2 + 4\left|a - b\right|^2}\right) \\
        \leq& \frac{1}{2} \left(2 \left|a_{12}\right|^2 + 2 \left|a + b\right|^2 + \left|b_{12}\right|^2 + \left|b_{21}\right|^2 + \frac{1}{2}\left(\left|b_{21}\right|+\left|b_{12}\right|\right)^2 + \frac{1}{2} \left(\left|b_{12}\right| - \left|b_{21}\right|\right)^2 + 2\left|a + b\right|^2\right) + \\
        & \frac{1}{2} \left(2 \left|a_{21}\right|^2 + 2 \left|a + b\right|^2 + \left|b_{12}\right|^2 + \left|b_{21}\right|^2 + \frac{1}{2}\left(\left|b_{21}\right|+\left|b_{12}\right|\right)^2 + \frac{1}{2} \left(\left|b_{12}\right| - \left|b_{21}\right|\right)^2 + 2\left|a + b\right|^2\right) + \\
        & \frac{1}{2} \left(2 \left|a_{12}\right|^2 + 2 \left|a - b\right|^2 + \left|b_{34}\right|^2 + \left|b_{43}\right|^2 + \frac{1}{2}\left(\left|b_{43}\right|+\left|b_{34}\right|\right)^2 + \frac{1}{2} \left(\left|b_{34}\right| - \left|b_{43}\right|\right)^2 + 2\left|a - b\right|^2\right) + \\
        & \frac{1}{2} \left(2 \left|a_{21}\right|^2 + 2 \left|a - b\right|^2 + \left|b_{34}\right|^2 + \left|b_{43}\right|^2 + \frac{1}{2}\left(\left|b_{43}\right|+\left|b_{34}\right|\right)^2 + \frac{1}{2} \left(\left|b_{34}\right| - \left|b_{43}\right|\right)^2 + 2\left|a - b\right|^2\right) \\
        =& 8|a|^2 + 8|b|^2 + 2(|a_{12}|^2 + |a_{21}|^2 + |b_{12}|^2 + |b_{21}|^2 + |b_{34}|^2 + |b_{43}|^2) = \frac{1}{2}.
\end{aligned}
\end{equation}

\section{Proof of Lemma \ref{lemma: dense}}\label{sec: dense proof}

In general, in order to prove that $\cY$ is dense in $\cX$, we prove that the closure of $\cY$ is the set $\cX$, which, according to the definition of closure, is equivalent to, for every $x \in \cX$ and $r>0$, the open ball centered at $x$ with radius equals to $r$ always contains at least an element in set $\cY$. In our scenarios, it means that for every $(A_x, B_x) \in \cX$ and $r>0$, we can always find a matrix pair $(A_y,B_y) \in \cY$ that satisfies the following conditions
    \begin{enumerate}
        \item $\norm{(A_y, B_y) - (A_x, B_x)}_2 \leq r$, \label{condition: 4}
        \item $\norm{A_y}_F^2 + \norm{B_y}_F^2 = \frac{1}{4}$, \label{condition: 5}
        \item $\tr A_y = \tr B_y = 0$, \label{condition: 6}
        \item $A_y$ and $B_y$ only possess distinct eigenvalues. \label{condition: 7}
    \end{enumerate}
    To facilitate the following analysis, we replace the 4 conditions above (\ref{condition: 4}, \ref{condition: 5}, \ref{condition: 6} and \ref{condition: 7}) by following conditions which are stricter
    \begin{enumerate}[resume]
        \item $\norm{A_y - A_x}_F \leq \frac{r}{\sqrt 2}$ and $\norm{B_y - B_x}_F \leq \frac{r}{\sqrt 2}$, \label{condition: 8}
        \item $\norm{A_y}_F = \norm{A_x}_F$ and $\norm{B_y}_F = \norm{B_x}_F$, \label{condition: 9}
        \item $\tr A_y = \tr B_y = 0$, \label{condition: 10}
        \item $A_y$ and $B_y$ only possess distinct eigenvalues. \label{condition: 11}
    \end{enumerate}
    Without loss of generality, we can only work on $A_y$ and apply the same method to $B_y$ to obtain a matrix pair $(A_y, B_y) \in \cY$.

    \vspace{7pt}
    Since singular values are unitarily invariant, according to Schur triangularization, we can always assume that $A_x$ is upper triangular
    \begin{align}
        A_x = \begin{bmatrix}
            a_{11}^{(x)}& * & * &*\\
		  0 & a_{22}^{(x)} & * & *\\
		  0 & 0 & a_{33}^{(x)} & *\\
		  0 & 0 & 0 & a_{44}^{(x)}\\
        \end{bmatrix}.
    \end{align}
    For more details, please refer to \cite{Horn_Johnson_2012}. 

    \vspace{7pt}
    There are three different cases
    \begin{enumerate}[label = (\roman*)]
        \item $A_x$ possesses 2 identical eigenvalues, \label{case: diag 4}
        \item $A_x$ possesses 3 identical eigenvalues, \label{case: diag 5}
        \item $A_x$ possesses 4 identical eigenvalues. \label{case: diag 6}
    \end{enumerate}
    \textbf{Initially, let us examine case \ref{case: diag 4}}, we can assume that $a_{11}^{(x)} = a_{22}^{(x)}$, then, we add a perturbation $\epsilon$ to $A_x$
    \begin{align}
        A_\epsilon = A_x + \epsilon = \begin{bmatrix}
            a_{11}^{(x)} + \epsilon_1 & * & * & * \\
		  0 & a_{22}^{(x)} - \epsilon_1 & * & * \\
		  0 & 0 & a_{33}^{(x)} & * \\
		  0 & 0 & 0 & a_{44}^{(x)} \\
        \end{bmatrix}, 
    \end{align}
    where
    \begin{align}
        \epsilon = \begin{bmatrix}
            \epsilon_1 & 0 & 0 & 0 \\
            0 & -\epsilon_1 & 0 & 0 \\
            0 & 0 & 0 & 0 \\
            0 & 0 & 0 & 0
        \end{bmatrix}.
    \end{align}
    Conditions \ref{condition: 10} and \ref{condition: 11} are automatically satisfied. Since we add a perturbation, $\norm{A_\epsilon}_F \neq \norm{A_x}_F$, we can multiply $A_\epsilon$ by a factor of $c = \frac{\norm{A_x}_F}{\norm{A_\epsilon}_F}$, and set $A_y = c A_\epsilon$ such that condition \ref{condition: 9} is satisfied. 
    \begin{align}
        \norm{A_y - A_x}_F &= \norm{A_y - A_\epsilon + A_\epsilon - A_x}_F \\
        &\leq \norm{(c - 1) A_\epsilon}_F + \norm{A_\epsilon - A_x}_F.
    \end{align}
    In order to satisfy condition \ref{condition: 8}, we can choose an arbitrarily small $\epsilon$ such that $c$ is infinitely close to $1$. Now, we have
    \begin{align}
        \forall r > 0, \exists \epsilon_1 \in \bbC \text{ such that } \norm{A_y - A_x}_F \leq \frac{r}{\sqrt 2}.
    \end{align}
    
    \textbf{Subsequently, we examine case \ref{case: diag 5}}, we can assume that $a_{11}^{(x)} = a_{22}^{(x)} = a_{33}^{(x)}$. The strategy is similar to case \ref{case: diag 4}. Firstly, we add a perturbation $\epsilon$
    \begin{align}
        A_\epsilon = A_x + \epsilon = \begin{bmatrix}
            a_{11}^{(x)} + \epsilon_1 & * & * & * \\
		  0 & a_{22}^{(x)} + \epsilon_2 & * & * \\
		  0 & 0 & a_{33}^{(x)} - \epsilon_1 - \epsilon_2 & * \\
		  0 & 0 & 0 & a_{44}^{(x)} \\
        \end{bmatrix}, 
    \end{align}
    where
    \begin{align}
        \epsilon = \begin{bmatrix}
            \epsilon_1 & 0 & 0 & 0 \\
            0 & \epsilon_2 & 0 & 0 \\
            0 & 0 & -\epsilon_1 - \epsilon_2 & 0 \\
            0 & 0 & 0 & 0
        \end{bmatrix}.
    \end{align}
    Then we set $A_y = \frac{\norm{A_x}_F}{\norm{A_\epsilon}_F} A_\epsilon$.

    \vspace{7pt}
    \textbf{At last, we examine case \ref{case: diag 6}}, we can assume that $a_{11}^{(x)} = a_{22}^{(x)} = a_{33}^{(x)} = a_{44}^{(x)}$. The strategy is also similar. We add a perturbation $\epsilon$
    \begin{align}
        A_\epsilon = A_x + \epsilon = \begin{bmatrix}
            a_{11}^{(x)} + \epsilon_1 & * & * & * \\
		  0 & a_{22}^{(x)} - \epsilon_1 & * & * \\
		  0 & 0 & a_{33}^{(x)} + \epsilon_2 & * \\
		  0 & 0 & 0 & a_{44}^{(x)} - \epsilon_2 \\
        \end{bmatrix}, 
    \end{align}
    where
    \begin{align}
        \epsilon = \begin{bmatrix}
            \epsilon_1 & 0 & 0 & 0 \\
            0 & -\epsilon_1 & 0 & 0 \\
            0 & 0 & \epsilon_2 & 0 \\
            0 & 0 & 0 & -\epsilon_2
        \end{bmatrix}.
    \end{align}
    Then we set $A_y = \frac{\norm{A_x}_F}{\norm{A_\epsilon}_F} A_\epsilon$.

    \vspace{7pt}
    At this point, we have proven that for every $(A_x, B_x) \in \cX$ we can always find $(A_y, B_y) \in \cY$ that is infinitely close to $(A_x, B_x)$ in cases \ref{case: diag 4}, \ref{case: diag 5} and \ref{case: diag 6}.

\bibliographystyle{unsrt}

\bibliography{citation}

\end{document}